\documentclass[11pt]{amsart}
\usepackage{amsthm}
\usepackage{amsmath,amstext,amsthm,amsfonts,amssymb,amsthm}
\usepackage{multicol}
\usepackage{latexsym}
\usepackage{color}
\usepackage{graphicx}
\usepackage{epsf}
\usepackage{epsfig}
\usepackage{epic}
\usepackage{graphics}
\usepackage{verbatim}
\usepackage{color}
\usepackage{hyperref}
\newtheorem{theorem}{Theorem}[section]
\newtheorem{lemma}[theorem]{Lemma}

\theoremstyle{definition}
\newtheorem{definition}[theorem]{Definition}
\newtheorem{proposition}[theorem]{Proposition}

\theoremstyle{remark}
\newtheorem{remark}[theorem]{Remark}
\numberwithin{equation}{section}
\newcommand{\abs}[1]{\lvert#1\rvert}

\newcommand{\HI}{\mathfrak{H}}

\newcommand{\R}{\mathbb{R}}

\newcommand{\D}{\mathcal{D}}
\newcommand{\C}{\mathbb{C}}

\newcommand{\W}{\mathcal{W}_{\eta}}
\newcommand{\B}{\mathcal{B}}

\newcommand{\qu}{\mathbf{q}}

\newcommand{\oqu}{\overline{q}}
\newcommand{\IV}{I_{V_\mathbb{H}^R}}

\begin{document}
\title[Continuous Frames]{S-spectrum and Associated Continuous Frames on  Quaternionic Hilbert Spaces}
\author{M. Khokulan$^1$, K. Thirulogasanthar$^2$, B. Muraleetharan$^1$}
\address{$^{1}$ Department of Mathematics and Statistics, University of Jaffna, Thirunelveli, Jaffna, Srilanka. }
\address{$^{2}$ Department of Computer Science and Software Engineering, Concordia University, 1455 de Maisonneuve Blvd. West, Montreal, Quebec, H3G 1M8, Canada.}
\email{mkhokulan@gmail.com, santhar@gmail.com, b.b.muraleetharan@gmail.com}
%\thanks{This research is part of an M.Phil thesis to be submitted to University of Jaffna }
\date{\today}
\begin{abstract}
As needed for the construction of rank $n$ continuous frames on a right quaternionic Hilbert space the so-called S-spectrum of a right quaternionic operator is studied. Using the S-spectrum, as for the case of  complex Hilbert spaces, along the lines of the arguments of  {\em Ann.Phys.}, {\bf 222} (1993), 1-37., various classes of rank $n$ continuous frames and their equivalencies on a right quaternionic Hilbert space are presented.
\end{abstract}
\maketitle
\pagestyle{myheadings}
%%%%%%%%%%%%%%%%%%%%%%%%%%%%%%%%%%%%%%%%%%%%%%%%%%%%%%%%%%%%%%%%%%%%%%%%%%%%%%%%%%%%%%%%%%%%%%%%%%%%%%%%%%%%%%%%%
\section{Introduction}
Frames were first introduced by Duffin and Schaeffer in a study of non-harmonic Fourier series \cite{DU}. However, among many others, the pioneering works of Daubechies et al. brought appropriate  attention to frames \cite{D1,D2}. Wavelets and coherent states of quantum optics are specific classes of continuous frames \cite{Alibk}. The study of frames has exploded in recent years, partly because of their applications in digital signal processing\cite{Ole, GR} and other areas of physical and engineering problems. In particular, they are an integral part of time-frequency analysis. It is crucial to find a specific class of frames to fit to a specific physical problem, because there is no universal class of frames that fits to all problems. As technology advances, physicists and engineers will face new problems and thereby our search for tools to solve them will continue.

A separable Hilbert space possesses an orthonormal basis and each vector in the Hilbert space can be uniquely written in terms of this orthonormal basis. Despite orthonormal bases are hard to find, this uniqueness restricted flexibility in applications and pleaded for an alternative. As a result frames entered to replace orthonormal bases. Frames are overcomplete classes of vectors in Hilbert spaces.  Thereby a vector in the Hilbert space can have infinitely many representations in terms of frame vectors. This redundancy of frames is the key to their success in applications. The role of redundancy varies according to the requirements of the application at hand. In fact, redundancy gives greater design flexibility which allows frames to be constructed to fit a particular problem in a manner not possible by a set of linearly independent vectors \cite{Alibk, Ole, D1,A1}.

Hilbert spaces can be defined over the fields $\mathbb{R}$, the set of all real numbers, $\mathbb{C}$, the set of all complex numbers, and $~\mathbb{H},~$ the set of all quaternions only \cite{Ad}. The fields $\mathbb{R}$ and $\mathbb{C}$ are associative and commutative and the theory of functional analysis is a well formed theory over real and complex Hilbert spaces. But the quaternions form a non-commutative associative algebra and this feature highly restricted mathematicians to work out a well-formed theory of functional analysis on quaternionic Hilbert spaces. Further, due to the noncommutativity there are two types of Hilbert spaces on quaternions, called right quaternion Hilbert space and left quaternion Hilbert space. In assisting the study of frames the functional analytic properties of the underlying Hilbert space are essential. In the sequel we shall investigate the necessary functional analytic properties as needed.

To the best of our knowledge a general theory of frames on quaternionic Hilbert spaces is not formulated yet. In this manuscript we shall construct rank $n$ continuous frames on a right quaternionic Hilbert space following the lines of \cite{A1}.  While the complex numbers are two dimensional the quaternions are four dimensional; the increase in the dimension is expected to give greater flexibility in applications. 

%%%%%%%%%%%%%%%%%%%%%%%%%%%%%%%%%%%%%%%%%%%%%%%%%%%%%%%%%%%%%%%%%%%%%%
\section{Mathematical preliminaries}
We recall few facts about quaternions, quaternionic Hilbert spaces and quaternionic functional calculus which may not be very familiar to the reader. For quaternions and quaternionic Hilbert spaces we refer the reader to \cite{Ad}. We shall use the recently introduced $S$-spectral calculus of quaternionic operators for which we refer to \cite{Ric, Fab, AF, AC}.
\subsection{Quaternions}
Let $\mathbb{H}$ denote the field of quaternions. Its elements are of the form $q=x_0+x_1i+x_2j+x_3k,~$ where $x_0,x_1,x_2$ and $x_3$ are real numbers, and $i,j,k$ are imaginary units such that $i^2=j^2=k^2=-1$, $ij=-ji=k$, $jk=-kj=i$ and $ki=-ik=j$. The quaternionic conjugate of $q$ is defined to be $\overline{q} = x_0 - x_1i - x_2j - x_3k$. Quaternions do not commute in general. However $q$ and $\oqu$ commute, and quaternions commute with real numbers. $|q|^2=q\oqu=\oqu q$ and $\overline{qp}=\overline{p}~\oqu.$ Since the quaternion field is measurable and locally compact we take a measure $d\mu$ on it. For instant $d\mu$ can be taken as a Radon measure or $d\mu=d\lambda d\omega$, where $d\lambda$ is a Lebesgue measure on $\C$ and $d\omega$  is a Harr measure on $SU(2)$. For details we refer the reader to, for example, \cite{Thi1} (page 12).

\subsection{Right Quaternionic Hilbert Space}
Let $V_{\mathbb{H}}^{R}$ be a linear vector space under right multiplication by quaternionic scalars (again $\mathbb{H}$ standing for the field of quaternions).  For $\phi ,\psi ,\omega\in V_{\mathbb{H}}^{R}$ and $q\in \mathbb{H}$, the inner product
$$\langle\cdot\mid\cdot\rangle:V_{\mathbb{H}}^{R}\times V_{\mathbb{H}}^{R}\longrightarrow \mathbb{H}$$
satisfies the following properties
\begin{enumerate}
\item[(i)]
$\overline{\langle\phi \mid \psi \rangle}=\langle \psi \mid\phi \rangle$
\item[(ii)]
$\|\phi\|^{2}=\langle\phi \mid\phi \rangle>0$ unless $\phi =0$, a real norm
\item[(iii)]
$\langle\phi \mid \psi +\omega\rangle=\langle\phi \mid \psi \rangle+\langle\phi \mid \omega\rangle$
\item[(iv)]
$\langle\phi \mid \psi q\rangle=\langle\phi \mid \psi \rangle q$
\item[(v)]
$\langle\phi q\mid \psi \rangle=\overline{q}\langle\phi \mid \psi \rangle$
\end{enumerate}
where $\overline{q}$ stands for the quaternionic conjugate. We assume that the
space $V_{\mathbb{H}}^{R}$ is complete under the norm given above. Then,  together with $\langle\cdot\mid\cdot\rangle$ this defines a right quaternionic Hilbert space, which we shall assume to be separable. Quaternionic Hilbert spaces share most of the standard properties of complex Hilbert spaces. In particular, the Cauchy-Schwartz inequality holds on quaternionic Hilbert spaces as well as the Riesz representation theorem for their duals.  Thus, the Dirac bra-ket notation
can be adapted to quaternionic Hilbert spaces:
$$\mid\phi q\rangle=\mid\phi \rangle q,\hspace{1cm}\langle\phi q\mid=\overline{q}\langle\phi \mid\;, $$
for a right quaternionic Hilbert space, with $\vert\phi \rangle$ denoting the vector $\phi$ and $\langle\phi \vert$ its dual vector. Let $A$ be an operator on a right quaternionic Hilbert space $V_\mathbb{H}^R$ with domain $V_\mathbb{H}^R$. The scalar multiple of $A$ should be written as $q A$ and the action must take the form \cite{Sha, TH}
\begin{equation}\label{act}
(qA)\mid\phi \rangle=(A\mid\phi \rangle)\overline{q}.
\end{equation}
Unless $q$ is real the scalar multiple $qA$ of a right linear operator is not necessarily right linear and several other usual properties of a scalar multiple of an operator may not hold \cite{Ric,Sha, Mu}.  The adjoint $A^{\dagger}$ of $A$ is defined as
\begin{equation}\label{Ad1}
\langle \psi \mid A\phi \rangle=\langle A^{\dagger} \psi \mid\phi \rangle;\quad\text{for all}~~~\phi ,\psi \in V_\mathbb{H}^R.
\end{equation}
An operator $A$ is said to be self-adjoint if $A=A^{\dagger}$. If $\phi \in V_\mathbb{H}^R\smallsetminus\{0\}$, then $|\phi \rangle\langle\phi |$ is a rank one projection operator. For operators $A, B$, by convention, we have 
\begin{equation}\label{Rank}
|A\phi\rangle\langle B\phi |=A|\phi \rangle\langle\phi |B^{\dagger}.
\end{equation}
Let $\D(A)$ denote the domain of $A$. $A$ is said to be right linear if
$$A(\phi q+\psi p)=(A\phi )q+(B\psi )p;\quad\forall\phi ,\psi \in\D(A), q, p\in \mathbb{H}.$$
The set of all right linear operators will be denoted by $\mathcal{L}(V_\mathbb{H}^R)$. For a given $A\in \mathcal{L}(V_\mathbb{H}^R)$, the range and the kernel will be
\begin{eqnarray*}
\mbox{ran}(A)&=&\{\psi \in V_\mathbb{H}^R~|~A\phi =\psi \quad\text{for}~~\phi \in\D(A)\}\\
\ker(A)&=&\{\phi \in\D(A)~|~A\phi =0\}.
\end{eqnarray*}
We call an operator $A\in \mathcal{L}(V_\mathbb{H}^R)$ bounded if
\begin{equation*}
\|A\|=\sup_{\|\phi \|=1}\|A\phi \|<\infty.
\end{equation*}
or equivalently, there exists $K\geq 0$ such that $\|A\phi \|\leq K\|\phi \|$ for $\phi \in\D(A)$. The set of all bounded right linear operators will be denoted by $\B(V_\mathbb{H}^R)$.
The following definition is the same as for the complex operators \cite{Hel}. In terms of this definition, we shall validate some properties of operators, using S-spectrum, as needed.
\begin{proposition}\label{SAD}
Let $A\in\B(V_\mathbb{H}^R)$. Then $A$ is self-adjoint if and only if for each $\phi\in V_{\mathbb{H}}^{R}$, $\langle A\phi\mid\phi\rangle\in\mathbb{R}$. 
\end{proposition}
\begin{proof}
A proof is given in \cite{Fa} for a left quaternionic Hilbert space $V_{\mathbb{H}}^{L}$, and it can be easily manipulated for $V_{\mathbb{H}}^{R}$. For a proof for the sufficient part one may also see proposition 2.17 (b) in \cite{Ric}. 
\end{proof}
\begin{definition}\label{Def1}
Let $A$ and $B$ be self-adjoint operators on $V_\mathbb{H}^R$. Then $A\leq B$ ($A$ less or equal to $B$) or equivalently $B\geq A$ if $\langle A\phi|\phi\rangle\leq\langle B\phi|\phi\rangle$ for all $\phi\in V_\mathbb{H}^R$. In particular $A$ is called positive if $A\geq 0.$
\end{definition}
\begin{theorem}\cite{Ric}\label{IT1}
Let $A\in\B(V_\mathbb{H}^R)$. If $A\geq 0$ then there exists a unique operator in $\B(V_\mathbb{H}^R)$, indicated by $\sqrt{A}=A^{1/2}$ such that $\sqrt{A}\geq 0$ and $\sqrt{A}\sqrt{A}=A$.
\end{theorem}
\begin{lemma}\label{L3}
Let $U_{\mathbb{H}}^{R}$ and $V_{\mathbb{H}}^{R}$ be right quaternion Hilbert spaces. Let $A:\D(A)\longrightarrow V_{\mathbb{H}}^{R}$ be a linear operator with domain $\D(A)\subseteq U_{\mathbb{H}}^{R} $ and $\mbox{ran}(A)\subseteq V_{\mathbb{H}}^{R},$ then the inverse $A^{-1}:\mbox{ran}(A)\longrightarrow \D(A) $ exists if and only if $A\phi=0\Rightarrow \phi=0.$
\end{lemma}
\begin{proof}Since the non-commutativity of quaternions does not play a role in the proof, it follows from its complex counterpart.
\end{proof}
For a positive bounded right quaternionic self-adjoint operator $A$ let
\begin{equation}\label{E7}
M(A)=\sup_{\|\phi\|=1}\langle\phi|A\phi\rangle
\quad\text{and}\quad
m(A)=\inf_{\|\phi\|=1}\langle\phi|A\phi\rangle.
\end{equation}
%%%%%%%%%%%%%%%%%%%
 \begin{lemma}\label{L2}
		 Let $A\in\B(V_\mathbb{H}^R)$ be a self-adjoint operator, then
 \begin{equation}\label{E13} 
 \left\|A\right\|=\sup_{\left\|\phi\right\|=1}\left|\left\langle \phi|A\phi \right\rangle\right|
 \end{equation} 		 
\end{lemma}
\begin{proof}
In fact, the non-commutativity of quaternions does not play a role in the proof, the proof follows from its complex counterpart. One may also see \cite{Fa} where a proof is given.
\end{proof}

%%%%%%%%%%%%%%%%%%%
Let $$GL(V_{\mathbb{H}}^{R})=\left\{A:V_{\mathbb{H}}^{R}\longrightarrow V_{\mathbb{H}}^{R}: A\mbox{~bounded and~} A^{-1}\mbox{~bounded}\right\}.$$
Then in a standard way one can see that $GL(V_{\mathbb{H}}^{R})$ is a group under the composition operation.
%%%%%%%%%%%%%%%%%%%%%%%%%%%%%%%%%%%%%%%%%%%%%%
\subsection{S-spectrum of a right quaternionic operator}
A consistent spectral theory for quaternionic operators is not so obvious. The problem in adapting the classical notion of spectrum to either the left or right quaternionic operators is well explained in \cite{AC,A1,Fab,Ric}. In the same references an appropriate notion of spectrum for quaternionic operators is introduced. This new-notion of spectrum is called the S-spectrum.  From \cite{AC,Ric} we shall extract the definition and some properties of S-spectrum for a right quaternionic operator as needed here. We shall also prove certain results, pertinent to the construction of frames, related to S-spectrum. To the best of our knowledge the results we prove do not appear in the literature. 
\begin{definition}\cite{Fab}\label{D1}
Let $A:V_\mathbb{H}^R\longrightarrow V_\mathbb{H}^R$ be a bounded right quaternionic linear operator. We define the S-spectrum $\sigma_S(A)$ of $A$ as$$\sigma_{S}(A)=\{\lambda\in \mathbb{H}: R_{\lambda}(A)=A^{2}-2Re(\lambda)A+|\lambda|^{2}\IV\mbox{~is not invertible in } \B(V_\mathbb{H}^R)\},$$
where $\lambda=\lambda_{0}+\lambda_{1}i+\lambda_{2}j+\lambda_{3}k$ is a quaternion, $Re(\lambda)=\lambda_{0},|\lambda|^{2}=\lambda_{0}^{2}+\lambda_{1}^{2}+\lambda_{2}^{2}+\lambda_{3}^{2}$ and $\IV$ is the identity operator on $V_\mathbb{H}^R$.
 The S-resolvant set $\rho_{S}(A)$ is defined by 
 $$\rho_{S}(A)=\mathbb{H}\setminus\sigma_{S}(A).$$
Or equivalently, the spherical resolvent set of $A$ is the set $\rho_{S}(A)\subset \mathbb{H}$ of quaternions $\lambda$ such that the three following conditions hold true:
\begin{itemize}
	\item [(a)]$\ker(R_{\lambda}(A))=\{0\}$
	\item [(b)]$\mbox{ran}(R_{\lambda}(A))$ is dense in $V_{\mathbb{H}}^{R}$
	\item [(c)]$R_{\lambda}(A)^{-1}:\mbox{ran}(R_{\lambda}(A))\longrightarrow D(A^{2})$ is bounded.
\end{itemize}
\end{definition}
The S-spectral radius of $A$ is defined as
$$r_S(A)=\sup\{|\qu|\in\R^+~|~\qu\in\sigma_S(A)\}.$$
\begin{proposition}\label{Pr1}\cite{AC, Ric} If $A$ is a bounded right linear quaternionic operator in $V_\mathbb{H}^R$, then the S-spectrum of $A$ is a compact non-empty subset of $\mathbb{H}$.
\end{proposition}
\begin{proposition}\label{Pr2}\cite{AC, Ric}Let $A\in\mathcal{L}(V_\mathbb{H}^R)$ and $A$ be self-adjoint, then $\sigma_S(A)\subseteq\mathbb{R}$.
\end{proposition} 
\begin{proposition}\label{Pr3}\cite{Ric} Let $A\in\B(V_\mathbb{H}^R)$, then $r_S(A)\leq\|A\|$. In particular, if $A$ is self-adjoint then $r_S(A)=\|A\|$.
\end{proposition}
\begin{proposition}\label{PP1}\cite{AC}
If $A\in\B(V_\mathbb{H}^R)$ is positive and self-adjoint, then $\sigma_S(A)\subseteq[0, \|A\|]$.
\end{proposition}
\begin{proposition}\label{P3}\cite{Ric}
If $A\in \B(V_{\mathbb{H}}^{R})$ and  $\lambda\in \mathbb{H} $ with $|\lambda|>\|A\|$, then $\lambda\in\rho_{S}(A).$
\end{proposition}
\begin{proof}It is another way of stating the result of Theorem 4.3 part(a) of \cite{Ric}.
\end{proof}
\begin{proposition}\label{LPr4}
Let $A$ be a positive self-adjoint operator. Then, for any $\phi,\psi \in V_{\mathbb{H}}^{R}$,
\begin{itemize}
\item[(a)]  $\mid \langle A\phi\mid \psi \rangle\mid^{2}\leq\langle A\phi\mid \phi\rangle\langle A\psi \mid \psi\rangle$ (Generalized Cauchy Inequality).
\item[(b)] $\|A\phi\|^{2}\leq\|A\|\langle A\phi\mid \phi\rangle$.
\end{itemize}
\end{proposition}
\begin{proof}
Let $\phi,\psi \in V_{\mathbb{H}}^{R}$. If $\psi \in\ker A$, as $A$ is self-adjoint, then (a) is obvious. Assume $\psi \notin \ker A$, then for any $\lambda,\mu\in \mathbb{H}$, as $A$ is positive and linear, we have
 $$0\leq\langle A(\phi\lambda -\psi \mu )\mid (\phi\lambda -\psi \mu )\rangle=\overline{\lambda}\langle A\phi\mid \phi\rangle\lambda+\overline{\mu}\langle A\psi \mid \psi \rangle\mu-\overline{\lambda}\langle A\phi\mid \psi \rangle\mu-\overline{\mu}\langle A\psi \mid \phi\rangle\lambda.$$ 
Choosing $\lambda=\langle A\psi \mid \psi \rangle$ and $\mu=\langle A\psi \mid \phi\rangle$, we obtain:
$$\langle A\psi \mid \psi \rangle[\langle A\phi\mid \phi\rangle\langle A\psi \mid \psi \rangle-\mid \langle A\phi\mid \psi \rangle\mid^{2}]\geq0$$ as $A$ is self-adjoint. Therefore, since $A$ is a positive operator, we have $$\mid \langle A\phi\mid \psi \rangle\mid^{2}\leq\langle A\phi\mid \phi\rangle\langle A\psi \mid \psi \rangle.$$
Let $\phi\in V_{\mathbb{H}}^{R}$. If $\phi\in \ker A$, then the result (b) is obvious. So assume that $\phi\notin \ker A$. Using (a) with $\psi =A\phi$, we have 
\begin{eqnarray*}
\|A\phi\|^{4}=\mid \langle A\phi\mid A\phi\rangle\mid^{2}&\leq&\langle A\phi\mid \phi\rangle\langle A^{2}\phi\mid A\phi\rangle\\
&\leq&\langle A\phi\mid \phi\rangle\|A^{2}\phi\|\,\|A\phi\|\\
&\leq&\langle A\phi\mid \phi\rangle\|A\|\,\|A\phi\|^{2}
\end{eqnarray*}
and the result (b) follows with a division by $\|A\phi\|^{2}$.
\end{proof}	
\begin{proposition}\label{rcl}
If $A\in\mathcal{B}(V_\mathbb{H}^R)$	and there exists $c>0$ such that $\|A\phi\|\geq c\|\phi\|$, for all $\phi\in V_{\mathbb{H}}^{R}$, then the range, ran$(A)$, of $A$ is closed.
\end{proposition}
\begin{proof}
Let $\{\psi _{n}\}$	be a sequence in ran$(A)$ such that $\psi_{n}\longrightarrow \psi $ as $n\longrightarrow\infty$. Then there exists a sequence $\{\phi_{n}\}$ such that $\psi _{n}=A\phi_{n}$, for all $n\in\mathbb{N}$. Using the fact that any convergence sequence is Cauchy, and 
$$\|\phi_{m}-\phi_{n}\|\leq\frac{1}{c}\|A(\phi_{m}-\phi_{n})\|=\frac{1}{c}\|\psi _{m}-\psi _{n}\|,$$ we have $\{\phi_{n}\}$ is Cauchy and $\phi_{n}\longrightarrow \phi$ as $n\longrightarrow\infty$ for some $\phi\in V_\mathbb{H}^R$. Hence $\psi =A\phi$ as $A$ is continuous. This concludes the result.
\end{proof}	
\begin{proposition}\label{inv}
 If $A\in GL(V_{\mathbb{H}}^{R})$, then $0\notin\sigma_{S}(A)$.
\end{proposition}
\begin{proof}
Suppose that, $0\in\sigma_S(A)$, then 
$A^{2}\phi=0,$ for some $\phi\in V_{\mathbb{H}}^{R}\smallsetminus\{0\}$. From this, we have $A\phi=0$, and $\phi\in\ker{A}\neq\{0\},$
or $A\phi\neq0$, and
$A\phi\in\ker{A}\neq\{0\}.$ Hence  $A\notin GL(V_{\mathbb{H}}^{R})$.
\end{proof}
\begin{proposition}\label{Pr5}
If $A\in\B(V_\mathbb{H}^R)$ and $A$ is self-adjoint, then
\begin{equation}\label{K1}
m(A)=\inf_{\|\phi\|=1}\langle A\phi|\phi\rangle=\max\{\alpha\in\R~|~\alpha I_{V_\mathbb{H}^R}\leq A\},
\end{equation}
\begin{equation}\label{K2}
M(A)=\sup_{\|\phi\|=1}\langle A\phi|\phi\rangle=\min\{\alpha\in\R~|~A\leq\alpha I_{V_\mathbb{H}^R}\},
\end{equation} and
\begin{equation}\label{K3}
\|A\|=\max\{|m(A)|,\abs{M(A)}\}
\end{equation}
\end{proposition}
\begin{proof}
The proof follows its complex counterparts. For the sake of completeness we provide the proof of (\ref{K1}). Since 
$\displaystyle\label{DmA}m(A)=\inf_{\|\phi\|=1}{\langle A\phi\mid\phi\rangle},
$
we have $m(A)\leq\langle A\phi\mid\phi\rangle,$ for all $\phi\in V_{\mathbb{H}}^{R}$ with $\|\phi\|=1$. Let $\psi\in V_{\mathbb{H}}^{R}$ with $\psi\neq 0$, then $\displaystyle m(A)\leq\left\langle A\frac{\psi}{\|\psi\|}\mid\frac{\psi}{\|\psi\|}\right\rangle.$ This implies that
$m(A)\|\psi\|^{2}\leq\langle A\psi\mid\psi\rangle,$
 or equivalently $\langle m(A)I_{V_\mathbb{H}^R}\;\psi\mid\psi\rangle\leq\langle A\psi\mid\psi\rangle.$
 Thus $m(A)I_{V_{\mathbb{H}}^{R}}\leq A~~\mbox{~and therefore~}$
\begin{equation}\label{5E1}
m(A)\in\{\alpha\in\R~|~\alpha I_{V_\mathbb{H}^R}\leq A\}.
\end{equation}
Now let $\alpha\in\{\alpha\in\R~|~\alpha I_{V_\mathbb{H}^R}\leq A\}$, then one can easily see that $\alpha\leq\langle A\phi\mid\phi\rangle,~~\mbox{~for all~}~~\phi\in V_{\mathbb{H}}^{R}~~\mbox{~with~}~~\|\phi\|=1$ and hence $\alpha\leq\inf_{\|\phi\|=1}\langle A\phi\mid\phi\rangle=m(A).$ This together with (\ref{5E1}) proves (\ref{K1}).
\end{proof}
\begin{proposition}\label{Prds}
	Let $A$ be a positive self-adjoint operator and $M(A)$ and $m(A)$ be as in (\ref{E7}). If $\lambda\notin[m(A),M(A)]$, then 
	\begin{itemize}
		\item[(a)] $R_{\lambda}(A)$ is one to one and ran$\,(R_{\lambda}(A))$ is closed.
		\item[(b)] ran$\,(R_{\lambda}(A))=V_{\mathbb{H}}^{R}$.
	\end{itemize}	
\end{proposition}
\begin{proof}
	Suppose that $\lambda\notin[m(A),M(A)]$ and $d=$dist$(\lambda,[m(A),M(A)])$. Let $\phi\in V_{\mathbb{H}}^{R}$ with $\|\phi\|=1$ and $\mu=\langle A\phi\mid\phi\rangle$, then $\langle (A-\mu \IV)\phi\mid\phi\rangle=0$.  Now using proposition (\ref{LPr4}), we have
	\begin{equation}\label{inq1}
	\|(A-\mu \IV)\phi\|^{2}\leq\|A-\mu \IV\|\,\langle (A-\mu\IV)\phi\mid\phi\rangle=0.
	\end{equation}
	Again using proposition (\ref{LPr4}) together with (\ref{inq1}), we get
	\begin{equation}\label{inq2}
	\|(A-\mu \IV)^{2}\phi\|^{2}\leq\|A-\mu \IV\|^{2}\,\|(A-\mu\IV)\phi\|^{2}=0.
	\end{equation}
	Now, (\ref{inq1}) and (\ref{inq2}) imply
	\begin{eqnarray*}
		\|R_{\lambda}(A)\phi\|&=& \|[(A-\mu \IV)-(\lambda \IV-\mu \IV)]^2\phi\|\\
		&=& \|(A-\mu \IV)^{2}\phi-2(A-\mu \IV)(\lambda \IV-\mu \IV)\phi+(\lambda \IV-\mu \IV)^{2}\phi\|\\
		&\geq&\mid\|(A-\mu \IV)^{2}\phi-[-(\lambda \IV-\mu \IV)^{2}\phi]\,\|-\|2(A-\mu \IV)(\lambda \IV-\mu \IV)\phi\|\mid\\
		&\geq& \mid\,\mid\|(A-\mu \IV)^{2}\phi\|-\mid \lambda-\mu\mid^{2}\,\mid-2\mid \lambda-\mu\mid\|(A-\mu \IV)\phi\|\,  \mid~\mbox{~~as~~}~\|\phi\|=1\\
		&=&\mid \lambda-\mu\mid^{2}\geq d^{2}>0.
	\end{eqnarray*}
	It follows that 
	\begin{equation}\label{ibd}
	\|R_{\lambda}(A)\psi\|\geq d^{2}\|\psi\|~\mbox{~~for all~~}~\psi\in V_{\mathbb{H}}^{R}.
	\end{equation}
	Hence $R_{\lambda}(A)$ is one to one and, by the proposition (\ref{rcl}), ran$\,(R_{\lambda}(A))$ is closed, this concludes the result (a).\\
	Assume that there exists $\varphi\in V_{\mathbb{H}}^{R}\smallsetminus\{0\}$ such that $\varphi\,\bot\,\mbox{ran}\,(R_{\lambda}(A))$. Then, since $A$ is self-adjoint, for all $\psi\in V_{\mathbb{H}}^{R}$, $$\langle(A-\lambda \IV)^{2}\psi\mid\varphi\rangle=\langle\psi\mid(A-\lambda \IV)^{2}\varphi\rangle=0.$$
	Therefore $R_{\lambda}(A)\varphi=(A-\lambda \IV)^{2}\varphi=0$ and it contradicts (\ref{ibd}). Thus $\mbox{ran}\,(R_{\lambda}(A))^{\bot}=\{0\}$ and so ran$\,(R_{\lambda}(A))$ is dense in $V_{\mathbb{H}}^{R}$. Hence the results follow.
\end{proof}
\begin{proposition}\label{Pr04}
	If $A$ is a positive self-adjoint operator and $M(A)$ and $m(A)$ are as in (\ref{E7}), then  $\sigma_S(A)\subset[m(A),M(A)]$.
\end{proposition}
\begin{proof}
	Suppose $\lambda\notin [m(A),M(A)]$. Then by the proposition (\ref{Prds}), $R_{\lambda}(A)^{-1}$ exists and for each $\psi\in V_{\mathbb{H}}^{R}$, there is a unique $\phi\in V_{\mathbb{H}}^{R}$ such that $\psi=R_{\lambda}(A)\phi$. We can easily see, from (\ref{ibd}) in the proof of proposition (\ref{Prds}), $$\|R_{\lambda}(A)^{-1}\psi\|=\|\phi\|\leq\frac{1}{d^{2}}\|\psi\|,$$ where $d=$dist$(\lambda,[m(A),M(A)])$. Thereby $R_{\lambda}(A)^{-1}\in\mathcal{B}(V_{\mathbb{H}}^{R})$. That is, $\lambda\notin\sigma_{S}(A)$. Hence the conclusion follows.
\end{proof}
\begin{proposition}\label{Pr4}
	If $A$ is a positive self-adjoint operator and $M(A)$ and $m(A)$ are as in (\ref{E7}), then $M(A), m(A)\in\sigma_S(A)$.
\end{proposition}
\begin{proof}
	From (\ref{K1}) in the proposition (\ref{Pr5}), we have for every $n\in\mathbb{N}$,  there exists $\psi_{n}\in V_{\mathbb{H}}^{R}\smallsetminus\{0\}$ with $\|\psi\|=1$ such that $$ m(A)+\frac{1}{n}>\langle A\psi_{n}\mid\psi_{n}\rangle$$ and hence we hold
	$$0\leq\langle(A-m(A)I_{V_\mathbb{H}^R})\psi_{n}\mid\psi_{n}\rangle<\frac{1}{n}.$$
	Therefore
	\begin{equation}\label{leq1}
	\lim_{n\rightarrow\infty}\langle(A-m(A)I_{V_\mathbb{H}^R})\psi_{n}\mid\psi_{n}\rangle=0.
	\end{equation}
	Now assume that $m(A)\in\rho_{S}(A)$. Then $R_{m}(A)=(A-m(A)I_{V_\mathbb{H}^R})^{2}$ is an invertible  bounded linear operator on $V_\mathbb{H}^R$. Also note that, since $A$ is self adjoint so are $R_m(A)$ and $A-m(A)I_{V_\mathbb{H}^R}$. Then for every $n\in\mathbb{N}$, using proposition (\ref{LPr4}), we have
	\begin{eqnarray*}
		1=\|\psi_{n}\|^{2}&=& \|R_m(A)^{-1}R_m(A)\psi_n\|^{2}\\
		&\leq & \|R_m(A)^{-1}\|\langle\psi_n|R_m(A)\psi_{n}\rangle\\
		&=& \|R_m(A)^{-1}\|\langle R_m(A)\psi_n|\psi_n\rangle\\
		&=&\|R_m(A)^{-1}\|\langle (A-m(A)I_{V_\mathbb{H}^R})^2\psi_n|\psi_n\rangle\\
		&=&\|R_m(A)^{-1}\|\langle (A-m(A)I_{V_\mathbb{H}^R})\psi_n|(A-m(A)I_{V_\mathbb{H}^R})\psi_n\rangle\\
		&=&\|R_m(A)^{-1}\|\|(A-m(A)I_{V_\mathbb{H}^R})\psi_n\|^2\\
		&\leq&\|R_m(A)^{-1}\|\|(A-m(A)I_{V_\mathbb{H}^R})\|\langle (A-m(A)I_{V_\mathbb{H}^R})\psi_n|\psi_n\rangle,\\
	\end{eqnarray*}
	which contradicts (\ref{leq1}). Hence $m(A)\in \sigma_S(A)$. By the same token 
	$$-M(A)=\max\{\alpha\in\R~|~\alpha I_{V_\mathbb{H}^R}\leq -A\}\in \sigma_S(-A).$$ That is, $(-A+M(A)I_{V_\mathbb{H}^R})^{2}=(A-M(A)I_{V_\mathbb{H}^R})^{2}$ is not invertible. Therefore, $M(A)\in\sigma_S(A)$  and this completes the proof.
\end{proof}
\begin{remark}\label{Re1}
If $A\in GL(V_\mathbb{H}^R)$ is positive and self-adjoint, then by propositions (\ref{PP1}), (\ref{inv}) and (\ref{Pr4}) we have
$
0<m(A),~M(A)<\infty.
$
\end{remark}
\begin{proposition}\label{Pr6}
Let $A\in\B(V_\mathbb{H}^R)$ and $A$ be self-adjoint, then
\begin{eqnarray}
m(A)&=&\min\{\lambda~|~\lambda\in\sigma_S(A)\}\label{sp1}\\
M(A)&=&\max\{\lambda~|~\lambda\in\sigma_S(A)\}\label{sp2}\\
\|A\|&=&\max\{\abs{\lambda}~|~\lambda\in\sigma_S(A)\}\label{sp3}
\end{eqnarray}
\end{proposition}
\begin{proof}
These results are immediate consequences of propositions (\ref{Pr4}), (\ref{Pr04}) and equation (\ref{K3}).
\end{proof}
\begin{proposition}\label{Pr7}
Let $A\in\B(V_\mathbb{H}^R)$ be a self-adjoint operator. Then,
\begin{enumerate}
\item[(a)] $A$ is positive if and only if $m(A)\geq 0$
\item[(b)] A is positive and $A^{-1}\in \mathcal{B}(V_{\mathbb{H}}^{R})$ if and only if $m(A)>0$.
\item[(c)]if $m(A)>0$, then $A^{-1}$ is a positive self adjoint operator, and $\min\sigma_S(A^{-1})=M(A)^{-1}$ and $\max\sigma_S(A^{-1})=m(A)^{-1}$
\end{enumerate}
\end{proposition}
\begin{proof}
For (a), suppose that $A$ is positive, then 
$\left\langle A\varphi|\varphi\right\rangle\geq0$ for all $\varphi\in V_{\mathbb{H}}^{R}.$
Thus $m(A)=\displaystyle\inf_{\|\varphi\|=1}\left\langle A\varphi|\varphi\right\rangle\geq 0.$ Conversely suppose that $m(A)\geq0$. Now let $\phi\in V_{\mathbb{H}}^{R}$ with $\phi\neq0$. If we take $\psi=\displaystyle\frac{\phi}{\left\|\phi\right\|}$, then $\|\psi\|=1$ and by the definition of $m(A)$, we get $\langle A\psi|\psi\rangle\geq 0$ and thereby $\langle A\phi|\phi\rangle\geq 0$. Hence $A$ is positive.For (b), suppose that $A$ is positive and invertible. Since $A$ is positive we have $m(A)\geq 0$. So it is enough to show that $m(A)>0$. On the contrary assume that $m(A)=0.$
Since $m(A)\in\sigma_{S}(A)$, $A^{2}-2m(A)A+m(A)^{2}I_{V_{\mathbb{H}}^{R}}$ is not invertible. That is, there exists $\varphi\neq0$ such that  $(A^{2}-2m(A)A+m(A)^{2}I_{V_{\mathbb{H}}^{R}})\varphi=0.$ Since $m(A)=0$, we have $A^2\varphi=0$. From the invertibility of $A$, it follows that
$\varphi=0$ which is a contradiction. Hence $m(A)>0$. 
Conversely suppose that $m(A)>0$. Therefore from (a), $A$ is positive. Now since $m(A)>0$, by (\ref{sp1}), it is easy to see that
$\lambda\geq m(A)>0$ for all $\lambda\in\sigma_{S}(A)$. Hence by the proposition (\ref{inv}), $A$ is invertible. Further, from (\ref{K1}) using Cauchy-Schwarz inequality, it directly follows that for all $\phi\in V_{\mathbb{H}}^{R}$, $\|A\phi\|\geq m(A)\|\phi\|$. This suffices to conclude that $A^{-1}\in\mathcal{B}(V_{\mathbb{H}}^{R})$.
For (c), suppose $m(A)>0$. By (b), $A$ is positive and invertible. Since $A$ is self-adjoint, clearly $A^{-1}$ is self-adjoint. From the existence of $A^{-1}$ we have that, for any given $\psi\in V_{\mathbb{H}}^{R}$, there exists $\phi\in V_{\mathbb{H}}^{R}$ such that $\phi=A^{-1}\psi.$ Then 
\begin{eqnarray*}
\langle A^{-1}\psi\mid \psi\rangle
&=&\langle \psi\mid (A^{-1})^{\dagger}\psi\rangle\\
&=&\langle \psi\mid A^{-1}\psi\rangle~~\mbox{~(as~}~~A^{-1}~~\mbox{~is self adjoint)}\\
&=&\langle A\phi\mid \phi\rangle\geq0~~\mbox{~(as~}~~A~~\mbox{~is positive)}.
\end{eqnarray*}
Thus $A^{-1}$ is positive. Since $M(A)\in \sigma_{S}(A)$,
$(A^{2}-2M(A)A+M(A)^{2}I_{V_{\mathbb{H}}^{R}})\phi=0,$
for some $\phi\neq0.$ 
Since $A$ is invertible, we get
$I_{V_{\mathbb{H}}^{R}}\phi-2M(A)A^{-1}\phi+M(A)^{2}(A^{-1})^{2}\phi=0.$
Hence
$$((A^{-1})^{2}-\frac{2}{M(A)} A^{-1}+\frac{1}{M(A)^{2}}I_{V_{\mathbb{H}}^{R}})\phi=0,$$
Therefore $\displaystyle M(A)^{-1}\in\sigma_{S}(A^{-1}).$
Hence
$M(A)^{-1}\geq\min\sigma_{S}(A^{-1}).$
On the other hand, let $\mu\in\sigma_{S}(A^{-1})$, then
$((A^{-1})^{2}-2\mu A^{-1}+\mu^{2}I_{V_{\mathbb{H}}^{R}})\phi=0,$
for some $\phi\neq0.$ Since $A^{-1}$ exists, we get
$(I_{V_{\mathbb{H}}^{R}}-2\mu A+\mu^{2}A^{2})\phi=0$
and
$\displaystyle(A^{2}-2\frac{1}{\mu}A+\frac{1}{\mu^{2}}I_{V_{\mathbb{H}}^{R}})\phi=0.$ 
Thus
$\displaystyle\frac{1}{\mu}\in\sigma_{S}(A).$
Hence from (\ref{sp2}), we have
$\frac{1}{\mu}\leq M(A).$
That is,
$M(A)^{-1}\leq \mu,$
for all $\mu\in\sigma_{S}(A^{-1})$. Thereby
$
M(A)^{-1}\leq\min\sigma_{S}(A^{-1})
$
and we conclude that
$M(A)^{-1}=\min\sigma_{S}(A^{-1}).
$
In a similar manner one can also prove that
$m(A)^{-1}=\max\sigma_{S}(A^{-1}).$
\end{proof}.
\begin{proposition}\label{Pr8}
Let (the class of unitary operators)  $$\mathcal{U}(V_{\mathbb{H}}^{R})\colon=\{U\colon V_{\mathbb{H}}^{R}\longrightarrow V_{\mathbb{H}}^{R}~\colon~UU^{\dagger}=U^{\dagger}U=I_{V_\mathbb{H}^R}\}.$$ Let $A\in GL(V_{\mathbb{H}}^{R})$, 
then $\sigma_{S}(A)=\sigma_{S}(UAU^{\dagger}),~~\mbox{~for all~}~~U\in\mathcal{U}(V_{\mathbb{H}}^{R}).$
\end{proposition}
\begin{proof}
One can easily obtain that $R_{\lambda}(UAU^{\dagger})=UR_{\lambda}(A)U^{\dagger}$ has an inverse in $\mathcal{B}(V_{\mathbb{H}}^{R})$ if and only if $R_{\lambda}(A)$ has such an inverse. Thus if $\lambda\notin\sigma_{S}(UAU^{\dagger})$ then $\lambda\notin\sigma_{S}(A)$, and vice versa. Hence
	$\sigma_{S}(A)=\sigma_{S}(UAU^{\dagger})$.
\end{proof}

%%%%%%%%%%%%%%%%%%%%%%%%%%%%%%%%%%%%%%%%%%%%%%%%%%%%%%%%%%%%%%%%%%%%
\section{Discrete frames on a finite dimensional right quaternion Hilbert space}
In a finite dimensional quaternion Hilbert space the discrete frame theory does not deviate from its complex counterpart. In fact the non-commutativity of quaternions plays a very little part if one follows the lines of the complex theory presented in \cite{Ole}. However, in duplicating it one should be aware of the conventions of a right quaternion Hilbert space. An interested reader can consult \cite{KH}, where we have verified  this fact. In this regard, we write down the definition and move to the continuous theory.
\begin{definition}\label{D3}
Let $V_\mathbb{H}^R$ be a finite dimensional right quaternion Hilbert space. A countable family of elements $\left\{\phi_{k}\right\}_{k\in I}$ in $V^{R}_{\mathbb{H}}$ is a frame for $V^{R}_{\mathbb{H}}$ if there exist constants $A,B>0$ such that
\begin{equation}\label{eq5}
A\left\|\phi\right\|^{2}\leq\displaystyle\sum_{k\in I}\left|\left\langle \phi|\phi_{k}\right\rangle\right|^{2}\leq B\left\|\phi\right\|^{2},\quad~\text{for all}~ \phi\in V^{R}_{\mathbb{H}}.
\end{equation}
\end{definition} 
The numbers $A$ and $B$ are called frame bounds. They are not unique. The \textit{optimal lower frame bound }is the supremum over all lower frame bounds, and the \textit{optimal upper frame bound} is the infimum over all upper frame bounds.  A frame is said to be normalized if $\left\|\phi_{k}\right\|=1,~\mbox{~~\text{for all}~} k\in I.$ 
%%%%%%%%%%%%%%%%%%%%%%%%%%%%%%%%%%%%%%%%%%%%%%%%%%%%%
\section{Continuous frames in a right quaternion Hilbert space}
In this section we shall investigate the rank $n$ continuous frames on a right quaternion Hilbert space. In particular, we shall follow the lines of arguments presented in \cite{A1} for the complex theory. The difficulty in transferring the complex theory to quaternions arises, besides the non-commutativity of quaternions, in the use of the spectral theory of operators.  Since we did not have a viable spectral theory for quaternion operators, this task appeared difficult so far.  However, the recently studied S-spectral theory for quaternion operators is promising. Using the S-spectrum of a positive bounded self-adjoint right quaternion operator, we present  rank $n$ continuous frames in the following.
\begin{theorem}\label{T2}
For each $q\in \mathbb{H},$ let the set $\{\eta_{q}^{i}: i=1,2,\cdots,n\}$ be  linearly independent in $V_{\mathbb{H}}^{R}.$ We define an operator $A$ by 
\begin{equation}\label{E3}
\sum_{i=1}^{n}\int_{\mathbb{H}}\left|\eta_{q}^{i}\right\rangle\left\langle \eta_{q}^{i}\right|d\mu(q)=A
\end{equation}
and we always assume that $A\in GL( V_{\mathbb{H}}^{R}).$ Then the operator $A$ is positive and self adjoint.
\end{theorem}
\begin{proof}
Let $\phi,\psi\in V_\mathbb{H}^R$ then
$A(\left|\phi\right\rangle)=\displaystyle\sum_{i=1}^{n}\int_{\mathbb{H}}\left|\eta_{q}^{i}\right\rangle\left\langle \eta_{q}^{i}|\phi\right\rangle d\mu(q).$
Since $\left\langle A\phi|\psi\right\rangle=(A\left|\phi\right\rangle)^{\dagger}\left|\psi\right\rangle,$ we have
\begin{eqnarray*} 
\left\langle A\phi|\psi\right\rangle&=&
 \left(\sum_{i=1}^{n}\int_{\mathbb{H}}\left|\eta_{q}^{i}\right\rangle\left\langle \eta_{q}^{i}|\phi\right\rangle d\mu(q)\right)^{\dagger}\left|\psi\right\rangle\\
 &=&\left(\sum_{i=1}^{n}\int_{\mathbb{H}}\left\langle \phi|\eta_{q}^{i}\right\rangle \left\langle \eta_{q}^{i}\right| d\mu(q)\right)\left|\psi\right\rangle\\
 &=&\sum_{i=1}^{n}\int_{\mathbb{H}}\left\langle \phi|\eta_{q}^{i}\right\rangle \left\langle \eta_{q}^{i}|\psi\right\rangle d\mu(q)=\left\langle \phi|A\psi\right\rangle.
\end{eqnarray*}
Therefore, $A$ is  self adjoint. Since, for $\phi\in V_\mathbb{H}^R,$
\begin{eqnarray*}
 \left\langle A\phi|\phi\right\rangle&=&
\sum_{i=1}^{n}\int_{\mathbb{H}}\left\langle \phi|\eta_{q}^{i}\right\rangle\overline{\left\langle \phi|\eta_{q}^{i}\right\rangle} d\mu(q)
 =\sum_{i=1}^{n}\int_{\mathbb{H}}\left|\left\langle \phi|\eta_{q}^{i}\right\rangle\right|^{2} d\mu(q)
 \geq 0,
 \end{eqnarray*} 
$A$ is positive.
 \end{proof}
 %%%%%%%%%%%%%%%%%%%%%%%%%%%%%%%%%%%%%%%%%%%%%%%%%%%%%%
\begin{theorem}\label{T4}
For $\phi\in V_\mathbb{H}^R$, we have
\begin{equation}\label{E23}
m(A)\|\phi\|^2\leq\sum_{i=1}^{n}\int_{\mathbb{H}}|\langle\eta_{q}^i|\phi\rangle|^2 d\mu(q)\leq M(A)\|\phi\|^2.
\end{equation}
\end{theorem}
\begin{proof}
 To prove (\ref{E23}), as in Theorem \ref{T2},
 we have
\begin{equation}\label{E24}
 \left\langle \phi|A\phi\right\rangle=
\sum_{i=1}^{n}\int_{\mathbb{H}}\left|\left\langle \phi|\eta_{q}^{i}\right\rangle\right|^{2} d\mu(q).
 \end{equation} 
From proposition (\ref{Pr5}) we have
\begin{eqnarray*}
m(A)I_{V_\mathbb{H}^R}\leq A\leq M(A)I_{V_\mathbb{H}^R}.
\end{eqnarray*}
Therefore, from (\ref{E24}), we get
\begin{equation*}
m(A)\|\phi\|^2\leq\sum_{i=1}^{n}\int_{\mathbb{H}}|\langle\eta_{q}^i|\phi\rangle|^2 d\mu(q)\leq M(A)\|\phi\|^2.
\end{equation*}
\end{proof}
The inequality (\ref{E23}) presents the frame condition for the set of vectors
$$\{\eta^i_q\in V_\mathbb{H}^R~|~i=1,2,\cdots, n,~q\in \mathbb{H}\}$$
with frame bounds $m(A)$ and $M(A)$. Now let us define a rank $n$ frame.
\begin{definition}\label{CFD1}
A set of vectors $\{\eta^i_q\in V_\mathbb{H}^R~|~i=1,2,\cdots, n,~q\in \mathbb{H}\}$ constitute a rank $n$ right quaternionic continuous frame, denoted by $F(\eta^i_q, A, n)$, if
\begin{enumerate}
\item[(i)] for each $q\in \mathbb{H}$, the set of vectors $\{\eta^i_q\in V_\mathbb{H}^R~|~i=1,2,\cdots, n\}$ is a linearly independent set.
\item[(ii)]there exists a positive operator $A\in GL(V_\mathbb{H}^R)$ such that
$$\sum_{i=1}^{n}\int_{\mathbb{H}}|\eta^i_q\rangle\langle \eta^i_q|d\mu(q)=A.$$
\end{enumerate}
\end{definition}
\begin{theorem}\label{CFT1}
Let  $F(\eta^i_q, A, n)$ be as in definition (\ref{CFD1}). Define $\overline{\eta}^i_q=A^{-1}\eta_q^i;\quad i=1,2,...,n,~q\in \mathbb{H}$. Then $F(\overline{\eta}_q^i, A^{-1}, n)$ forms a rank $n$ quaternionic continuous frame and it is called the dual frame of  $F(\eta^i_q, A, n)$.
\end{theorem}
\begin{proof}
From (\ref{Rank}) it is clear that
\begin{equation}\label{E25}
\sum_{i=1}^{n}\int_{\mathbb{H}}\left|\overline{\eta}_{q}^{i}\right\rangle\left\langle \overline{\eta}_{q}^{i}\right|d\mu(q)=A^{-1}.
 \end{equation} 
From proposition (\ref{Pr7}) we have
\begin{eqnarray*}
M(A)^{-1}I_{V_\mathbb{H}^R}\leq A^{-1}\leq m(A)^{-1}I_{V_\mathbb{H}^R},
\end{eqnarray*}
therefore we readily obtain the frame condition,
for $\phi\in V_\mathbb{H}^R$, 
\begin{equation}\label{E26}
M(A)^{-1}\|\phi\|^2\leq\sum_{i=1}^{n}\int_{\mathbb{H}}|\langle\overline{\eta}_{q}^i|\phi\rangle|^2 d\mu(q)\leq m(A)^{-1}\|\phi\|^2.
\end{equation}
\end{proof}
\begin{remark}\hfill
\begin{enumerate}
\item[$\bullet$]The quantity $\displaystyle\mathcal{W}(F(\eta^i_q, A, n))=\frac{M(A)-m(A)}{M(A)+m(A)}$ is called the width or snugness of the frame $F(\eta^i_q, A, n)$. Clearly $0\leq\mathcal{W}(F(\eta^i_q, A, n))<1$, and $\mathcal{W}(F(\eta^i_q, A, n))$, in a manner, measures the S-spectral width of the frame operator $A$.
\item[$\bullet$]A frame $F(\eta^i_q, A, n)$ is called tight if $\mathcal{W}(F(\eta^i_q, A, n))=0$, that is $A=M(A)I_{V_\mathbb{H}^R},$ and it is called self-dual if $A=I_{V_\mathbb{H}^R}$.
\item[$\bullet$]A frame $F(\eta^i_q, A, n)$ and its dual $F(\overline{\eta}_q^i, A^{-1}, n)$ have the same frame width.
\end{enumerate}
\end{remark}
Associated naturally to any frame $F(\eta^i_q, A, n)$ there is a self-dual tight frame. We present it next.
\begin{proposition}\label{CPr1}
Let $F(\eta^i_q, A, n)$ be a frame and let $\hat{\eta}_q^i=A^{-1/2}\eta_q^i$, then $F(\hat{\eta}^i_q, I_{V_\mathbb{H}^R}, n)$ is a tight frame and it is self-dual.
\end{proposition}
\begin{proof}
From (\ref{Rank}) it is clear that
\begin{equation}\label{E27}
\sum_{i=1}^{n}\int_{\mathbb{H}}\left|\hat{\eta}_{q}^{i}\right\rangle\left\langle \hat{\eta}_{q}^{i}\right|d\mu(q)=I_{V_\mathbb{H}^R}
 \end{equation} 
and thereby, for $\phi\in V_\mathbb{H}^R$,
$\sum_{i=1}^{n}\int_{\mathbb{H}}|\langle\hat{\eta}_{q}^i|\phi\rangle|^2 d\mu(q)=\|\phi\|^2.$
\end{proof}
\begin{proposition}\label{T-frame}
Let $F(\eta^i_q, A, n)$ be a frame and $T\in GL(V_\mathbb{H}^R)$. If we take $\tilde{\eta}_q^i=T\eta_q^i$, then $F(\tilde{\eta}_q^i, TAT^{\dagger},n)$ is a frame.
\end{proposition}
\begin{proof}
From (\ref{Rank}) it is clear that
\begin{equation}\label{Ex28}
\sum_{i=1}^{n}\int_{\mathbb{H}}\left|\tilde{\eta}_{q}^{i}\right\rangle\left\langle \tilde{\eta}_{q}^{i}\right|d\mu(q)=TAT^{\dagger}.
 \end{equation} 
and the frame condition follows from the frame condition of $A$ and the boundedness of $T$.
\end{proof}
In the above proposition when $T$ is unitary we obtain an interesting new class of frames. We present it next.
\begin{theorem}\label{CFT2}
Let $U\in GL(V_\mathbb{H}^R)$ be a unitary operator and  $F(\eta^i_q, A, n)$ be a frame as in definition (\ref{CFD1}). Define $\tilde{\eta}_q^i=U\eta_q^i;~~i=1,2,...,n$ and $\tilde{A}=UAU^{\dagger}$, then $F(\tilde{\eta}_q^i, \tilde{A}, n)$ is a rank $n$ continuous frame. The frames $F(\eta^i_q, A, n)$ and $F(\tilde{\eta}_q^i, \tilde{A}, n)$ are said to be unitarily equivalent and  $\mathcal{W}(F(\eta^i_q, A, n))=\mathcal{W}(F(\tilde{\eta}_q^i, \tilde{A}, n)).$
\end{theorem}
\begin{proof}
From (\ref{Rank}) it is clear that
\begin{equation}\label{E28}
\sum_{i=1}^{n}\int_{\mathbb{H}}\left|\tilde{\eta}_{q}^{i}\right\rangle\left\langle \tilde{\eta}_{q}^{i}\right|d\mu(q)=\tilde{A}.
 \end{equation} 
From proposition (\ref{Pr8}) we have
\begin{eqnarray*}
m(A)I_{V_\mathbb{H}^R}\leq \tilde{A}\leq M(A)I_{V_\mathbb{H}^R},
\end{eqnarray*}
therefore we readily obtain the frame condition,
for $\phi\in V_\mathbb{H}^R$, 
\begin{equation}\label{E29}
m(A)\|\phi\|^2\leq\sum_{i=1}^{n}\int_{\mathbb{H}}|\langle\tilde{\eta}_{q}^i|\phi\rangle|^2 d\mu(q)\leq M(A)\|\phi\|^2.
\end{equation}
The equality of the frame widths is obvious.
\end{proof}
In proposition (\ref{CPr1}) we obtained a self-dual tight frame $F(\hat{\eta}_q^i, I_{V_\mathbb{H}^R}, n)$ associated to the frame $F(\eta_q^i, A, n)$. Is proposition (\ref{CPr1}) the only way to obtain a self-dual tight frame from a frame $F(\eta_q^i, A, n)$? We answer this in the next proposition.
\begin{proposition}\label{CPr2}
Let $U\in GL(V_\mathbb{H}^R)$ be a unitary operator and $F(\eta^i_q, A, n)$ be a frame as in definition (\ref{CFD1}). Define $\hat{\eta}_q^i=UA^{-1/2}\eta_q^i;~~i=1,2,..,n$, then $F(\hat{\eta}_q^i, I_{V_\mathbb{H}^R}, n)$ is a self-dual tight frame.
\end{proposition}
\begin{proof}
Since $U$ is unitary and $A^{-1/2}$ is self-adjoint, from (\ref{Rank}), we have
\begin{equation}\label{E30}
\sum_{i=1}^{n}\int_{\mathbb{H}}\left|\hat{\eta}_{q}^{i}\right\rangle\left\langle \hat{\eta}_{q}^{i}\right|d\mu(q)=UA^{-1/2}A(UA^{-1/2})^{\dagger}=UU^{\dagger}=I_{V_\mathbb{H}^R},
 \end{equation} 
and therefore  $F(\hat{\eta}_q^i, I_{V_\mathbb{H}^R}, n)$ is  self-dual and tight.
\end{proof}
Consider a frame $F(\eta_q^i, A, n)$. For each $q\in \mathbb{H}$ consider the positive operator
\begin{equation}\label{E31}
S(q)=\sum_{i=1}^n\mid\eta_q^i\rangle\langle\eta_q^i\mid
\end{equation}
then the frame operator is $\displaystyle\int_{\mathbb{H}}S(q)d\mu(q)=A.$ A natural question arises. Is there more than one set of linearly independent vectors for which (\ref{E31}) is satisfied? Indeed, for each $q\in \mathbb{H}$, the choice of linearly independent set $\{\eta_q^i\}_{i=1}^n$ is as large as $\mathcal{U}(n, \mathbb{H})$, the set of all $n\times n$ quaternionic unitary matrices.
\begin{theorem}\label{CFT3}
Let $F(\eta^i_q, A, n)$ be a frame and $S(q)$ is as in (\ref{E31}). Then $\{\xi_q^i\}_{i=1}^n$ is a linearly independent set of vectors for which 
\begin{equation}\label{E32}
\displaystyle S(q)=\sum_{i=1}^n\mid\xi_q^i\rangle\langle\xi_q^i\mid\end{equation}
 if and only if there exists a unitary matrix $u(q)=(u(q)_{ij})_{n\times n}\in\mathcal{U}(n,\mathbb{H})$ such that
\begin{equation}\label{E33}
\displaystyle\xi_q^i=\sum_{j=1}^n\eta_q^ju(q)_{ji};~~i=1,2,...,n.
\end{equation}
\end{theorem}
\begin{proof}
From the bra-ket convention on a right quaternionic Hilbert space we have
$$|\eta_q^ju(q)_{ji}\rangle\langle\eta_q^ku(q)_{ki}|=|\eta_q^j\rangle u(q)_{ji}\overline{u(q)_{ki}}\langle\eta_q^k|$$
and from the unitarity of $u(q)$ we also have $\displaystyle\sum_{j=1}^nu(q)_{ij}\overline{u(q)_{kj}}=\delta_{ik}$. Therefore, if (\ref{E33}) holds true then we have (\ref{E32}). Conversely suppose that $\{\xi_q^i\}_{i=1}^n$ is a linearly independent set and (\ref{E32}) holds. Let $\phi\in V_\mathbb{H}^R$, then from (\ref{E32}) and (\ref{E31}) we have
$$\langle\phi|S(q)\phi\rangle=\sum_{i=1}^n|\langle\phi|\xi_q^i\rangle|^2
=\sum_{i=1}^n|\langle\phi|\eta_q^i\rangle|^2.$$
If we take $q_i'=\langle\xi_q^i|\phi\rangle\in \mathbb{H}$ and $q_i=\langle\eta_q^i|\phi\rangle\in \mathbb{H}$, then we can write
\begin{equation}\label{E031}
\sum_{i=1}^n|q_i'|^2=\sum_{i=1}^n|q_i|^2.
\end{equation}
Let $S(q)V_\mathbb{H}^R=\hat{\HI}_\mathbb{H}$ and $P(q):V_\mathbb{H}^R\longrightarrow\hat{\HI}_\mathbb{H}$ be the projection operator. Then $P(q)V_\mathbb{H}^R$ can be spanned by both set of vectors $\{\xi_q^i\}_{i=1}^n$ and $\{\eta_q^i\}_{i=1}^n$. That is, $\{\xi_q^i\}_{i=1}^n$ and $\{\eta_q^i\}_{i=1}^n$ are two different bases for $P(q)V_\mathbb{H}^R$. Therefore, there exists an invertible matrix $u(q)=(u(q)_{ij})_{n\times n}$ such that $\xi_q^i=\displaystyle\sum_{j=0}^n\eta_q^ju(q)_{ji}$. Thus, for $\phi\in V_\mathbb{H}^R$, we have
$$\langle\phi|\xi_q^i\rangle=\sum_{j=1}^n\langle\phi|\eta_q^j\rangle u(q)_{ji}=q'_i.$$
Therefore, by (\ref{E031}), $u(q)$ is unitary.
\end{proof}
%%%%%%%%%%%%%%%%%%%%%%%%%%%%%%%%%%%%%%%%%%%%%%%%%%%%%%%%%%%%%%%%%%%%%%%%%%%%%%%
\subsection{Frames and reproducing kernels}
In this subsection we shall investigate links between rank $n$ continuous frames and reproducing kernels. Using those links we shall also obtain equivalencies between various rank $n$ continuous frames.
Let $$\HI_\mathbb{H}=L^2_{\mathbb{H}^n}(\mathbb{H},\mu)=\left\{f:\mathbb{H}\longrightarrow \mathbb{H}^n~:~f(q)=(f_1(q),...,f_n(q)),~\sum_{i=1}^{n}\int_{\mathbb{H}}|f_i(q)|^2d\mu(q)<\infty\right\}$$
be a right quaternionic Hilbert space.
\begin{proposition}\label{RP1}
Let $V_\mathbb{H}^R$ be a separable Hilbert space and $F(\eta_q^i,A,n)$ be a rank $n$ frame. Then
\begin{equation}\label{RE1}
\W:V_\mathbb{H}^R\longrightarrow \HI_\mathbb{H}\quad\mbox{defined by}\quad(\W\phi)_i(q)=\langle\eta_q^i|\phi\rangle;\quad i=1,2,\cdots, n
\end{equation}
is a bounded linear map.
\end{proposition}
\begin{proof}
For $\phi_1,\phi_2\in V_\mathbb{H}^R$ and $\alpha,\beta\in \mathbb{H}$, we have
$$\langle\eta_q^i|\phi_1\alpha+\phi_2\beta\rangle=\langle\eta_q^i|\phi_1\rangle\alpha+
\langle\eta_q^i|\phi_2\rangle\beta;\quad i=1,2,\cdots,n.$$
For $\phi\in V_\mathbb{H}^R$, from the frame condition, we have
$$\|\W\phi\|_{{\HI}_\mathbb{H}}^2=\sum_{i=1}^{n}\int_{\mathbb{H}}|\langle\eta_q^i|\phi\rangle|^2d\mu(q)\leq M(A)\|\phi\|^2.$$
Therefore $\W$ is linear and bounded.
\end{proof}
\begin{proposition}\label{RP2}
Let $\W:V_\mathbb{H}^R\longrightarrow{\HI}_\mathbb{H}$ be as in (\ref{RE1}). Then the image $\HI_{\eta}=\W(V_\mathbb{H}^R)$ is a closed subspace of ${\HI}_\mathbb{H}$ and therefore $\HI_{\eta}$ is itself a Hilbert space.
\end{proposition}
\begin{proof}
From the theorem (4.2), for each $\phi\in V_{\mathbb{H}}^{R}$, we have $$m(A)\|\phi\|^2\leq\|\W\phi\|_{\HI_\mathbb{H}}^2,$$ and hence from proposition (2.13) the result follows.
\end{proof}
Note that $\W^{-1}$ exists on the range $\HI_{\eta}$ of $\W$. An inner product and a norm on $\HI_{\eta}$ is given in the following proposition, (\ref{RP3}).
\begin{proposition}\label{RP03}
Let $\W$ be as in (\ref{RE1}), $F(\eta_q^i, A, n)$ be a rank $n$ continuous frame in $V_\mathbb{H}^R$ and $A_{\eta}=\W A\W^{-1}$. Then
\begin{itemize}
\item[(a)] $A=\W^{\dagger}\W$.
\item[(b)] $A_{\eta}$ is self-adjoint and positive.
\end{itemize}
\end{proposition}
\begin{proof}
For any given $\phi,\psi\in V_{\mathbb{H}}^{R}$, we have 
\begin{eqnarray*}
\langle \W^{\dagger}\W\phi|\psi\rangle
&=&\langle \W\phi|\W\psi\rangle\\
&=&\sum_{i=1}^{n}\int_{\mathbb{H}}\overline{(\W\phi)_{i}(q)}(\W\psi)_{i}(q)\,d\mu(q)\\
&=&\sum_{i=1}^{n}\int_{\mathbb{H}}\langle\phi|\eta_q^i\rangle\langle\eta_q^i|\psi\rangle d\mu(q)\\
&=&\langle \phi|A\psi\rangle=\langle A\phi|\psi\rangle.
\end{eqnarray*}
Hence (a) follows. To prove the result (b), from the proposition (\ref{SAD}), it suffices to prove that, for each $\Phi\in \mathfrak{H}_\mathbb{H}$, $\langle A_{\eta}\Phi\mid\Phi\rangle_{\HI_\mathbb{H}}$ is real. Now for each $\Phi\in\mathfrak{H}_\mathbb{H}$,
\begin{eqnarray*}
\langle A_{\eta}\Phi|\Phi\rangle_{\HI_\mathbb{H}}
&=&\langle \W A\W^{-1}\Phi|\Phi\rangle_{\HI_\mathbb{H}}\\
&=&\langle A\W^{-1}\Phi|\W^{\dagger}\Phi\rangle\\
&=&\langle \W^{\dagger}\Phi|\W^{\dagger}\Phi\rangle
~~\mbox{~~as~~}~~A=\W^{\dagger}\W\\
&=& \|\W^{\dagger}\Phi\|^2\in\mathbb{R}.
\end{eqnarray*}
This concludes the proof.
\end{proof}
\begin{proposition}\label{RP3}
Let $\W$ be as in (\ref{RE1}), $F(\eta_q^i, A, n)$ be a rank $n$ continuous frame in $V_\mathbb{H}^R$ and $A_{\eta}=\W A\W^{-1}$. Then, for $\Phi,\Psi\in\HI_{\eta}$,
\begin{equation}\label{RE2}
\langle\Phi|\Psi\rangle_{\eta}=\langle\Phi|A_{\eta}^{-1}\Psi\rangle_{\HI_{\eta}}
\end{equation}
is an inner product on $\HI_{\eta}$ and if we take $\|\Phi\|_{\eta}=\langle\Phi|\Phi\rangle_{\eta}$, then $\|\cdot\|_{\eta}$ is a norm on $\HI_{\eta}$.
\end{proposition}
\begin{proof}
Let $\Phi,\Psi\in\HI_{\eta}$,  then
	\begin{eqnarray*}
		\overline{\langle \Phi|\Psi \rangle_{\eta}}
		&=&\overline{\langle\Phi|A_{\eta}^{-1}\Psi\rangle_{\HI_\mathbb{H}}}
		=\langle A_{\eta}^{-1}\Psi|\Phi\rangle_{\HI_\mathbb{H}}
		=\langle \Psi|(A_{\eta}^{-1})^{\dagger}\Phi\rangle_{\HI_\mathbb{H}}\\
		&=&\langle \Psi|A_{\eta}^{-1}\Phi\rangle_{\HI_\mathbb{H}}~~\mbox{~~as~~}  A_{\eta}^{-1} ~~\mbox{~~is self-adjoint~~}\\
		&=&\langle \Psi|\Phi\rangle_{\eta}.     
	\end{eqnarray*}
Therefore, we have the conjugate symmetry. Linearity in the second argument and positive-definiteness can be easily verified from direct calculations. Hence the result holds. 
\end{proof}
\begin{proposition}\label{Isem}
Let $\W:V_\mathbb{H}^R\longrightarrow{\HI}_\mathbb{H}$ be as in (\ref{RE1}). Then $\W$ is an isometric embedding.
\end{proposition}
\begin{proof}
Theorem (\ref{T4}) provides sufficient condition to have the injectivity and continuity of $\W$.
Since, for each $\phi\in V_{\mathbb{H}}^{R}$, as $A=\W^{\dagger}\W,$ we have
\begin{eqnarray*}
\|\W\phi\|_{\eta}^{2}
&=&\langle\W\phi|A_{\eta}^{-1}\W\phi\rangle_{\HI_\mathbb{H}}
=\langle\phi|\W^{\dagger}A_{\eta}^{-1}\W\phi\rangle
=\langle\phi|\phi\rangle=\|\phi\|^2.
\end{eqnarray*}
Thus $\W$ is an isometry.
\end{proof}
\begin{remark}
By definition, the operator $A_{\eta}$ is an isomorphism.
\end{remark}
\begin{proposition}\label{RP4}
Let the assumptions of proposition (\ref{RP3}) hold, then
\begin{enumerate}
\item[(a)] $\|\cdot\|_{\eta}$ and $\|\cdot\|_{{\HI}_\mathbb{H}}$ are equivalent norms,
\item[(b)] $\HI_{\eta}$ is closed with respect to the norm $\|\cdot\|_{\eta}$,
\item[(c)] $(\HI_{\eta}, \|\cdot\|_{\eta})$ is a Hilbert space,
\end{enumerate}
\end{proposition}
\begin{proof}
Let $\Phi\in \HI_{\eta}$. Then, by the Cauchy-Schwarz inequality, we have
\begin{equation}\label{neq1}
\|\Phi\|_{\eta}^{2}
=\langle\Phi|A_{\eta}^{-1}\Phi\rangle_{\HI_\mathbb{H}}\leq\|A_{\eta}^{-1}\|\,\|\Phi\|_{\HI_\mathbb{H}}^{2}.
\end{equation}
Since there is a $\Psi\in\HI_\mathbb{H}$ such that $\Phi=A_{\eta}\Psi$, by the proposition (\ref{LPr4}), we have
$$\|A_{\eta}\Psi\|_{\HI_\mathbb{H}}^{2}\leq\|A_{\eta}\|_{\HI_\mathbb{H}}\langle A_{\eta}\Psi|\Psi\rangle_{\HI_\mathbb{H}}.$$ This implies that
\begin{equation}\label{neq2}
\frac{1}{\|A_{\eta}\|_{\HI_\mathbb{H}}}\|\Phi\|_{\HI_\mathbb{H}}^{2}\leq\langle \Phi|A_{\eta}^{-1}\Phi\rangle_{\HI_{\mathbb{H}}}=\|\Phi\|_{\eta}^{2}.
\end{equation}
The inequalities (\ref{neq1}) and (\ref{neq2}) prove (a).
Part (a) together with the proposition (\ref{RP2}) proves (b) and (c). 
\end{proof}
\begin{proposition}\label{RP5}
Let $F(\eta_{q}^i, A, n)$ be a rank $n$ continuous frame on $V_\mathbb{H}^R$. Consider the map $K^{\eta}:\mathbb{H}\times \mathbb{H}\longrightarrow M_n(\mathbb{H})$, where $M_n(\mathbb{H})$ is the set of all $n\times n$ quaternionic matrices, defined by $K^{\eta}(p,q)=(K_{ij}^{\eta}(p,q))_{n\times n}$ with $K^{\eta}_{ij}(p,q)=\langle\eta_p^i|A^{-1}\eta_q^j\rangle;~~i,j=1,2,...,n.$ Then $K^{\eta}(p,q)$ is a reproducing kernel on $\HI_{\eta}$ and $\HI_{\eta}$ is the corresponding reproducing kernel Hilbert space.
\end{proposition}
\begin{proof}
Using the definition of dual frame operator $A^{-1}$ (see theorem \ref{CFT1})), for all $p,q,r\in \mathbb{H},~i,j=1,2,...,n$, we can easily verify
\begin{enumerate}
\item[(a)] $K_{ii}^{\eta}(q,q)>0$,
\item[(b)] $K_{ij}^{\eta}(p,q)=\overline{K_{ji}^{\eta}(q,p)}$,
\item[(c)]$\displaystyle\sum_{k=1}^n\int_{\mathbb{H}}K_{ik}^{\eta}(p,q)K_{kj}^{\eta}(q,r)d\mu(q)=K_{ij}^{\eta}(p,r)$
\end{enumerate}
The property (c) is called the reproducing property of the kernel.
\end{proof}
Note that the reproducing property (c) of proposition (\ref{RP5}) of the kernel has the effect of acting as the evaluation map for any vector $\Phi\in\HI_{\eta}$.
\begin{proposition}\label{RP6}
The map $E_{\eta}^i(q):\HI_{\eta}\longrightarrow \mathbb{H}$ given by
\begin{equation}\label{RE3}
E_{\eta}^i(q)\Phi=\sum_{j=1}^n\int_{\mathbb{H}}K_{ij}^{\eta}(q,p)\Phi_j(p)d\mu(p)=\Phi_i(q)
\end{equation}
is linear and it is an evaluation map.
\end{proposition}
\begin{proof}
For $\Phi,\Psi\in\HI_{\eta}$ and $\alpha,\beta\in \mathbb{H}$, it is straightforward that $E_{\eta}^i(q)(\Phi\alpha+\Psi\beta)=E_{\eta}^i(q)\Phi\alpha+E_{\eta}^i(q)\Psi\beta.$  It is an evaluation map because for a fixed $q\in\mathbb{H}$, the map $E_{\eta}^i(q)$ sends each vector $\Phi\in\mathfrak{H}_{\eta}$ to the value $\Phi_i(q)\in\mathbb{H}$.
\end{proof}
So far we have established the existence of a reproducing kernel corresponding to a quaternionic rank $n$ continuous frame $F(\eta_q^i,A,n)$ and the corresponding reproducing kernel Hilbert space $\HI_{\eta}$. Let us summerize it as follows.
\begin{definition}\label{RD1}
Let  $F(\eta_q^i,A,n)$ be a rank $n$ continuous frame in $V_\mathbb{H}^R$, then the $n\times n$ matrix valued function $K^{\eta}:\mathbb{H}\times \mathbb{H}\longrightarrow M_n(\mathbb{H})$ defined by
\begin{equation}\label{RE4}
K^{\eta}(p,q)=\left(K_{ij}^{\eta}(p,q)\right)_{n\times n}=\left(\langle\eta_p^i|A^{-1}\eta_q^j\rangle\right)_{n\times n}
\end{equation}
is a reproducing kernel corresponding to the given frame and it is called the frame kernel.
\end{definition}
Now we are in a position to establish various frame equivalencies using the associated reproducing kernels.
\begin{theorem}\label{RT1}
In $V_\mathbb{H}^R$, a frame and its dual frame have the same reproducing kernel.
\end{theorem}
\begin{proof}
Let $F(\eta_q^i,A,n)$ be a frame in $V_\mathbb{H}^R$ and $F(\overline{\eta}_q^i,A^{-1},n)$ is its dual frame, that is $\overline{\eta}_{q}^i=A^{-1}\eta_q^i$. Then for all $p,q\in \mathbb{H},~i,j=1,2,...,n$, as $A$ is self-adjoint, we have
$$K_{ij}^{\eta}(p,q)=\langle\eta_p^i|A^{-1}\eta_q^j\rangle
=\langle AA^{-1}\eta_p^i|A^{-1}\eta_q^j\rangle=\langle A^{-1}\eta_p^i|AA^{-1}\eta_q^j\rangle
=\langle\bar{\eta}_p^i|A\bar{\eta}_q^j\rangle
=K_{ij}^{\bar{\eta}}(p,q).$$
Therefore $K^{\eta}=K^{\bar{\eta}}.$
\end{proof}
\begin{theorem}\label{RT2}
Let $F(\eta_q^i,A,n)$ be a frame in $V_\mathbb{H}^R$. Define $F(\xi_q^i, \tilde{A},n)$ with $$\displaystyle\xi_q^i=\sum_{j=1}^n\eta_q^ju(q)_{ji};~i=1,2,...,n,$$ where for $q\in \mathbb{H}, u(q)=(u(q)_{ij})_{n\times n}$ is a unitary matrix as in (\ref{E33}). Then
\begin{equation}\label{RE5}
K^{\xi}(p,q)=u(p)^{\dagger}K^{\eta}(p,q)u(q).
\end{equation}
\end{theorem}
\begin{proof}
From theorem (\ref{CFT3}) we have $\tilde{A}=A$. From the bra-ket convention of a right quaternionic Hilbert space, (\ref{RE5}) follows as
$$K_{ij}^{\xi}(p,q)=\langle\xi_p^i|\tilde{A}^{-1}\xi_q^j\rangle
=\sum_{k,l=1}^n\overline{u_{ki}(p)}\langle\eta_p^k|A^{-1}\eta_q^l\rangle u_{lj}(q)=
\sum_{k,l=1}^n\overline{u_{ki}(p)}K_{kl}^{\eta}(p,q) u_{lj}(q).$$
\end{proof}
\begin{definition}\label{RD2}
Two frames $F(\eta_q^i,A,n)$ and $F(\xi_q^i, \tilde{A},n)$ in $V_\mathbb{H}^R$ are said to be gauge equivalent if their reproducing kernels are related by (\ref{RE5}) and $\tilde{A}=A$.
\end{definition}
\begin{theorem}\label{RT3}
Let $F(\eta_q^i,A,n)$ and $F(\xi_q^i, \tilde{A},n)$ be gauge equivalent frames in $V_\mathbb{H}^R$. Then the reproducing kernel Hilbert space $\HI_{\xi}$ for the frame $F(\xi_q^i, \tilde{A},n)$ consists of vectors of the type
\begin{equation}\label{RE6}
\Psi(q)=u(q)^{\dagger}\Phi(q);~~q\in \mathbb{H},~~u(q)=(u(q)_{ij})_{n\times n} ~~{\text{ is a unitary matrix as in (\ref{E33}),}}
\end{equation}
and $\Phi\in\HI_{\eta}$, the reproducing kernel Hilbert space of the frame $F(\eta_q^i,A,n)$.
\end{theorem}
\begin{proof}
From proposition (\ref{RP4}) we have isometries $\W:V_\mathbb{H}^R\longrightarrow \HI_{\eta}$ and $\mathcal{W}_{\xi}:V_\mathbb{H}^R\longrightarrow\HI_{\xi}$. Let $\Psi\in\HI_{\xi}$, then there exists $\phi\in V_\mathbb{H}^R$ such that $\mathcal{W}_{\xi}\phi=\Psi$. Further $\phi\in V_\mathbb{H}^R$ implies $\W\phi=\Phi\in\HI_{\eta}$, therefore
$$\Psi_i(q)=(\mathcal{W}_{\xi}\phi)_i(q)=\langle\xi_q^i|\phi\rangle
=\sum_{j=1}^n\overline{u(q)_{ji}}\langle\eta_q^i|\phi\rangle=\sum_{j=1}^n\overline{u(q)_{ji}}\Phi_j(q)
=(u(q)^{\dagger}\Phi(q))_i.$$
\end{proof}
\begin{theorem}\label{RT4}
Let $F(\tilde{\eta}_q^i, \tilde{A},n)$ and $F(\eta_q^i,A,n)$ be unitary equivalent frames in $V_\mathbb{H}^R$, then $K^{\eta}(p,q)=K^{\tilde{\eta}}(p,q)$ and $\HI_{\eta}=\HI_{\tilde{\eta}}$.
\end{theorem}
\begin{proof}
Since the frames are unitary equivalent, there is a unitary operator $U\in GL(V_\mathbb{H}^R)$ such that $\tilde{\eta}_q^i=U\eta_q^i$ and $\tilde{A}=UAU^{\dagger}$. Hence
$$K^{\tilde{\eta}}_{ij}(p,q)=\langle\tilde{\eta}_p^i|\tilde{A}^{-1}\tilde{\eta}_q^j\rangle
=\langle U\eta_p^i|UA^{-1}U^{\dagger}U\eta_q^j\rangle=\langle\eta_p^i|A^{-1}\eta_q^i\rangle=K^{\eta}_{ij}(p,q).$$
From the isometries $\W:V_\mathbb{H}^R\longrightarrow\HI_{\eta}$ and $\mathcal{W}_{\tilde{\eta}}:V_\mathbb{H}^R\longrightarrow \HI_{\tilde{\eta}}$, for $\Phi\in\HI_{\tilde{\eta}}$, for some $\phi\in V_\mathbb{H}^R$, we have
$$\Phi_i(q)=(\mathcal{W}_{\tilde{\eta}}\phi)_i(q)
=\langle\tilde{\eta}_q^i|\phi\rangle=\langle U\eta_q^i|\phi\rangle=\langle \eta_q^i|U^{\dagger}\phi\rangle=\Psi_i(q),$$
for some $\Psi\in\HI_{\eta}$ as $U^{\dagger}\phi\in V_\mathbb{H}^R$. Thus $\HI_{\tilde{\eta}}\subseteq\HI_{\eta}$ and the other inclusion follows similarly.
\end{proof}
\begin{definition}\label{RD3}
Two frames are said to be kernel equivalent in $V_\mathbb{H}^R$ if their kernels are gauge related in the sense of (\ref{RE5}).
\end{definition}
From the above definition gauge equivalent frames are kernel equivalent. However, the converse not necessarily be true. The following theorem somewhat answers the converse.
\begin{theorem}\label{RT5}
Let $T\in GL(V_\mathbb{H}^R)$ and $u(q)=(u(q)_{ij})_{n\times n}~~{\text{ be a quaternion unitary matrix}}$. For two frames $F(\tilde{\eta}_q^i, \tilde{A}, n)$ and $F(\eta_q^i,A, n)$ the relations
\begin{equation}\label{RE7}
\tilde{\eta}_q^i=\sum_{j=0}^n(T\eta_q^i)u(q)_{ji}~~~{\mbox and}~~\tilde{A}=TAT^{\dagger}
\end{equation}
 hold if and only if the frames are kernel equivalent.
\end{theorem}
\begin{proof}
Suppose (\ref{RE7}) holds. Also note that $(T^{\dagger})^{-1}=(T^{-1})^{\dagger}$ (see Theorem 2.15 in \cite{Ric}). Then
\begin{eqnarray*}
& &K_{ij}^{\tilde{\eta}}(p,q)=\langle\tilde{\eta}_p^i|\tilde{A}^{-1}\tilde{\eta}_q^j\rangle
=\sum_{k,l=1}^n\langle (T\eta_p^k)u(p)_{ki}\mid (T^{\dagger})^{-1}A^{-1}T^{-1}(T\eta_q^l)u(q)_{lj}\rangle\\
&=&\sum_{k,l=1}^n\overline{u(p)_{ki}}\langle (T\eta_p^k)\mid (T^{\dagger})^{-1}A^{-1}T^{-1}(T\eta_q^l)\rangle u(q)_{lj}
=\sum_{k,l=1}^n\overline{u(p)_{ki}}\langle \eta_p^k\mid A^{-1}\eta_q^l\rangle u(q)_{lj},
\end{eqnarray*}
and therefore $K^{\tilde{\eta}}=u(p)^{\dagger}K^{\eta}(p,q)u(q)$. Conversely suppose that the frames are kernel equivalent. Therefore, there are $n\times n$ quaternion unitary matrices $u(p)$ and $u(q)$ such that $K^{\tilde{\eta}}=u(p)^{\dagger}K^{\eta}(p,q)u(q)$ holds, and hence
\begin{eqnarray*}
\tilde{\eta}_q^j&=&\sum_{i=1}^{n}\int_{\mathbb{H}}\mid\tilde{\eta}_p^i\rangle\langle
\tilde{\eta}_p^i\mid\tilde{A}^{-1}\tilde{\eta}_q^j\rangle d\mu(p)
=\sum_{i,k,l=1}^n\int_{\mathbb{H}}|\tilde{\eta}_p^i\rangle\overline{u(p)_{ki}}\langle \eta_p^k|A^{-1}\eta_q^l\rangle u(q)_{lj}d\mu(p)\\
&=&\sum_{l=1}^n(T\eta_q^l)u(q)_{lj},
\end{eqnarray*}
where
$\displaystyle T=\sum_{i,k=1}^n\int_{\mathbb{H}}|\tilde{\eta}_p^i\rangle\overline{u(p)_{ki}}\langle \eta_p^k|A^{-1}d\mu(p)$
is, obviously, a bounded operator on $V_\mathbb{H}^R$. As $u(q)$ is unitary, we have
\begin{eqnarray*}
\tilde{A}&=&\sum_{i=1}^{n}\int_{\mathbb{H}}|\tilde{\eta}_q^i\rangle\langle\tilde{\eta}_q^i|d\mu(q)
=\sum_{i,k,l=1}^n\int_{\mathbb{H}}|(T\eta_q^l)u(q)_{li}\rangle\langle (T\eta_q^k)u(q)_{ki}|d\mu(q)\\
&=&\sum_{i,l,k=1}^n T\int_{\mathbb{H}}|\eta_q^l\rangle u(q)_{li}\overline{u(q)_{ki}}\langle\eta_q^k|d\mu(q)T^{\dagger}
=T\sum_{i=1}^{n}\int_{\mathbb{H}}|\eta_q^i\rangle\langle\eta_q^i|d\mu(q)T^{\dagger}
=TAT^{\dagger}.
\end{eqnarray*}
Now let us show that $T$ is invertible. For $\Phi,\Psi\in\mathfrak{H}_K$, the reproducing kernel Hilbert space $(\mathfrak{H}_K=\mathfrak{H}_{\eta}=\mathfrak{H}_{\tilde{\eta}})$, and $\phi,\psi\in V_\mathbb{H}^R$, set $\mathcal{W}_{\eta}\phi=\Phi_{\eta}$, $\mathcal{W}_{\eta}\psi=\Psi_{\eta}$, then
\begin{eqnarray*}
\langle\phi|T\psi\rangle&=&
\sum_{i,k=1}^n\int_{\mathbb{H}}\langle\phi|\tilde{\eta}_p^i\rangle \overline{u(p)_{ki}}\langle\eta_p^k| A^{-1}\psi\rangle d\mu(p)\\
&=&
\sum_{i,k=1}^n\int_{\mathbb{H}} \overline{(\mathcal{W}_{\tilde{\eta}}\phi)_i(p)}~~\overline{u(p)_{ki}}(\mathcal{W}_{\eta}A^{-1}\psi)_k(p) d\mu(p),
\end{eqnarray*}
where we have used 
$$\langle\phi|\tilde{\eta}_p^i\rangle=\overline{(\mathcal{W}_{\tilde{\eta}}\phi)_i(p)}
\quad\text{and}\quad
\langle\eta_p^k|A^{-1}\psi\rangle=(\mathcal{W}_{\eta}A^{-1}\psi)_k(p).$$
Since $A_{\eta}=\mathcal{W}_{\eta}A\mathcal{W}_{\eta}^{-1}$ we have $A_{\eta}^{-1}\mathcal{W}_{\eta}=\W A^{-1}$. Further, we have$(\mathcal{W}_{\tilde{\eta}}\phi)_i(p)=(\Phi_{\tilde{\eta}})_i(p)$ and $(A_{\eta}^{-1}\W \psi)_k=(A_{\eta}^{-1}\Psi_{\eta})_k(p)$. Therefore
\begin{eqnarray*}
\langle\phi|T\psi\rangle&=&\sum_{i,k=1}^n\int_{\mathbb{H}} \overline{(\mathcal{W}_{\tilde{\eta}}\phi)_i(p)}~~\overline{u(p)_{ki}}(A_{\eta}^{-1}\mathcal{W}_{\eta}\psi)_k(p) d\mu(p)\\
&=&\sum_{i,k=1}^n\int_{\mathbb{H}} \overline{(\Phi_{\tilde{\eta}})_i(p)}~~\overline{u(p)_{ki}}(A_{\eta}^{-1}\Psi_{\eta})_k(p)d\mu(p).
\end{eqnarray*}
Since $(\Phi_{\tilde{\eta}})_i(p)$ and $\displaystyle\sum_{k=1}^n\overline{u(p)_{ki}}(A_{\eta}^{-1}\Psi_{\eta})_k(p);~~i=1,2,...,n$ define vectors in $L^2_{\mathbb{H}^n}(\mathbb{H}, d\mu)$, the last integral converges. Further $u(p)$ is unitary, therefore from the expression of $\langle\tilde{\eta}_p^i|\tilde{A}^{-1}\tilde{\eta}_q^j\rangle$ we get
$$\langle\eta_p^i|A^{-1}\eta_q^j\rangle
=\sum_{k,l=1}^n u(p)_{ik}\langle\tilde{\eta}_p^k|\tilde{A}^{-1}\tilde{\eta}_q^l\rangle\overline{u(q)_{lj}}.$$
By setting
$$T'=\sum_{i,k=1}^n\int_{\mathbb{H}}|\eta_p^i\rangle u(p)_{ik}\langle\tilde{\eta}_p^k|\tilde{A}^{-1}d\mu(p)$$
we see that
$$T'\tilde{\eta}_q^l=\sum_{i,k=1}^n\int_{\mathbb{H}}|\eta_p^i\rangle u(p)_{ik}\langle\tilde{\eta}_p^k|\tilde{A}^{-1}\tilde{\eta}_q^l\rangle d\mu(p).$$
Therefore
\begin{eqnarray*}
\sum_{l=1}^n(T'\tilde{\eta}_q^l)\overline{u(q)_{lj}}&=&
\sum_{i,k,l=1}^n\int_{\mathbb{H}}|\eta_p^i\rangle u(p)_{ik}\langle\tilde{\eta}_p^k|\tilde{A}^{-1}\tilde{\eta}_q^l\rangle\overline{u(q)_{lj}} d\mu(p)\\
&=&\sum_{i=1}^n|\int_{\mathbb{H}}\eta_p^i\rangle\langle\eta_p^i|A^{-1}\eta_q^j\rangle d\mu(p)=\eta_q^j.
\end{eqnarray*}
That is, we have
\begin{equation}\label{In}
\tilde{\eta}_q^j=\sum_{l=1}^n(T\eta_q^l)u(q)_{lj}\quad\text{and}\quad \eta_q^j=\sum_{l=1}^n(T'\tilde{\eta}_q^l)\overline{u(q)_{lj}}.
\end{equation}
Since $\{\eta_q^i~|~i=1,2,...,n;~q\in \mathbb{H}\}$ and $\{\tilde{\eta}_q^i~|~i=1,2,...,n;~q\in \mathbb{H}\}$ are total in $V_\mathbb{H}^R$, (\ref{In}) implies that $T'=T^{-1}$.
%%%%%%%%%%%%%%%%%%%
%%%%%%%%%%%%%%%%%%
%%%%%%%%%%%%%%%%%%%
\end{proof}
Let $\mathfrak{F}$ be the set of all rank-$n$ continuous frames in $V_\mathbb{H}^R$. Define a relation $\sim$ on $\mathfrak{F}$ by
$$F_1\sim F_2\Longleftrightarrow F_1~~\text{and}~~F_2~~\text{are kernel equivalent frames}.$$
Then it can easily be verified that $\sim$ is an equivalent relation. For any fixed $F_0\in\mathfrak{F}$, let $[F_0]$ denotes the equivalent class of $F_0$.
\begin{definition}
An equivalent class $[F_0]$ is said to be self-dual if for a given $F_1\in[F_0]$ there exists $F_2\in[F_0]$ such that $F_2$ is the dual of $F_1$.
\end{definition}
\begin{proposition}\label{Pr9}
Every class of kernel equivalent frames is self-dual.
\end{proposition}
\begin{proof}
Fix $F_0(\eta_q^i, A, n)$ and set $\overline{\eta}_q^i=A^{-1}\eta_q^i$. Then by theorem (\ref{CFT1}),$ F(\overline{\eta}_q^i, A^{-1}, n)$ is dual to $F_0$. The kernel equivalence follows as
$$K_{ij}^{\overline{\eta}}(p,q)=\langle\overline{\eta}_p^i|A\overline{\eta}_q^j\rangle
=\langle A^{-1}\eta_p^i|\eta_q^i\rangle=\langle\eta_p^i|A^{-1}\eta_q^j\rangle=K_{ij}^{\eta}(p,q).$$
That is $K^{\overline{\eta}}(p,q)=u(p)^{\dagger}K^{\eta}(p,q)u(q)$ with $u(p)=u(q)=\mathbb{I}_n$, the $n\times n$ identity matrix.
\end{proof}
\begin{theorem}\label{CFT4}
Each kernel equivalence class contains a unique subclass of self-dual tight frames which can be generated by the joint action, on any fixed member of this subclass, of the set of all $n\times n$ unitary matrices, $\mathcal{U}(n,\mathbb{H})$, and the set of all unitary operators, $\mathcal{U}(V_\mathbb{H}^R)$.
\end{theorem}
\begin{proof}
Fix $F(\eta_q^i, A, n)$ a rank-$n$ frame. Let $\hat{\eta}_q^i=A^{-\frac{1}{2}}\eta_q^i$, $u(q)\in \mathcal{U}(n,\mathbb{H})$, $U\in \mathcal{U}(V_\mathbb{H}^R)$, and 
\begin{equation}\label{X}
\xi_q^i=\sum_{j=0}^n(U\hat{\eta}_q^j)u(q)_{ji}.
\end{equation}
Then
\begin{eqnarray*}
\sum_{i=1}^{n}\int_{\mathbb{H}}|\xi_q^i\rangle\langle\xi_q^i|d\mu(q)&=&
\sum_{i,j,k=1}^n\int_{\mathbb{H}}|U\hat{\eta}_q^j\rangle u(q)_{ji}\overline{u(q)_{ki}}\langle U\hat{\eta}_q^k|d\mu(q)\\
&=&U\sum_{i=1}^{n}\int_{\mathbb{H}}|A^{-1/2}\eta_q^i\rangle\langle A^{-1/2}\eta_q^i|U^{\dagger}d\mu(q)=\IV
\end{eqnarray*}
and
\begin{eqnarray*}
K_{ij}^{\xi}(p,q)&=&\langle\xi_p^i|\xi_q^j\rangle=\sum_{k,l=1}^n
\langle (U\hat{\eta}_q^k)u(q)_{ki}|(U\hat{\eta}_q^l)u(q)_{lj}\rangle\\
&=&\sum_{k,l=1}^n
\overline{u(q)_{ki}} \langle(U\hat{\eta}_q^k)|(U\hat{\eta}_q^l)\rangle u(q)_{lj}
=\sum_{k,l=1}^n
\overline{u(q)_{ki}}\langle \hat{\eta}_q^k|\hat{\eta}_q^l\rangle u(q)_{lj}\\
&=&\sum_{k,l=1}^n
\overline{u(q)_{ki}}\langle \eta_q^k|A^{-1}\eta_q^l\rangle u(q)_{lj}=\sum_{k,l=1}^n
\overline{u(q)_{ki}}K_{kl}^{\eta}(p,q)u(q)_{lj}.
\end{eqnarray*}
Therefore $\mathfrak{S}=\{F(\xi_q^i, \IV, n)~|~\xi_q^i$~~\text{is given by (\ref{X})}\}$\subset [F(\eta_q^i, A,n)]$ and the class $\mathfrak{S}$ is self-dual and tight.
\end{proof}
\begin{definition}\label{CFD2}
Two frames $F(\eta_q^i, A, n)$ and $F(\tilde{\eta}_q^i, \tilde{A}, n)$ are said to be bundle equivalent if there exists a rank-$n$ operator $T(q),~q\in \mathbb{H}$ on $V_\mathbb{H}^R$ such that
\begin{equation}\label{B}
\tilde{\eta}_q^i=\sum_{j=1}^n (T(q)\eta_q^j)u(q)_{ji}.
\end{equation}
\end{definition}
Note that in the above definition we do not impose any connection between $A$ and $\tilde{A}$. In this regard, this kind of equivalence is different from kernel equivalence. In fact, we do not necessarily have any connection between the corresponding reproducing kernels.
\begin{theorem}\label{CFT5}
Two bundle equivalent frames $F(\eta_q^i, A, n)$ and $F(\tilde{\eta}_q^i, \tilde{A}, n)$ are kernel equivalent if there exists $T\in GL(V_\mathbb{H}^R)$ such that
$T(q)=TP(q)\quad\text{and}\quad \tilde{A}=TAT^{\dagger},$
where $P(q)$ is the projection operator onto the range of $S(q)$, the operator as in equation (\ref{E32}).
\end{theorem}
\begin{proof}
Consider
\begin{eqnarray*}
K_{ij}^{\tilde{\eta}}(p,q)&=&\langle\tilde{\eta}_p^i|\tilde{A}^{-1}\tilde{\eta}_q^j\rangle
=\sum_{k,l=1}^n\overline{u(p)_{ki}}\langle TP(p)\eta_p^k|(T^{\dagger})^{-1}A^{-1}T^{-1}TP(q)\eta_q^l\rangle u(q)_{lj}\\
&=&\sum_{k,l=1}^n\overline{u(p)_{ki}}\langle P(p)\eta_p^k|A^{-1}P(q)\eta_q^l\rangle u(q)_{lj}
=\sum_{k,l=1}^n\overline{u(p)_{ki}}\langle \eta_p^k|A^{-1}\eta_q^l\rangle u(q)_{lj}
\end{eqnarray*}
as $P(q)$ is a projection operator and $\{\eta_q^i~|~i=1,2,...,n\}$ is a linearly independent set.
\end{proof}
\subsection*{Acknowledgement} MK would like to thank S. Srisatkunarajah for his interest and encouragement. 

\end{document}